\documentclass[runningheads, envcountsame, a4paper]{llncs}

\usepackage{amsmath}
\usepackage{tikz}
\usetikzlibrary{arrows,automata,shapes,shapes.multipart,decorations.markings,positioning}
\usepackage{calc}
\usepackage{wrapfig}

\newcommand{\outputs}{\mathsf{outputs}}
\newcommand{\sset}[2]{\left\{~#1  \mid
      \begin{array}{l}#2\end{array}
         \right\}}
\newcommand{\networkVars}{\mathtt{V}}
\newcommand{\usedVars}{\mathtt{U}}

\newcounter{sncolumncounter}
\newcounter{snrowcounter}

\def \nodelabel#1{%
\setcounter{snrowcounter}{1}
 \foreach \i in {#1}{%
   \draw (\sncolwidth*\value{sncolumncounter},\value{snrowcounter}) node[anchor=south]{\i};
   \addtocounter{snrowcounter}{1}
 }
\addtocounter{snrowcounter}{-1}
 \addtocounter{sncolumncounter}{1}
}

\newcommand{\sncolwidth}{0.7} 

\def \addcomparator#1#2{%
    \draw (\sncolwidth*\value{sncolumncounter},#1) node[circle,fill=black,minimum size=4pt,inner sep=0pt,outer sep=0pt]{}--(\sncolwidth*\value{sncolumncounter},#2) node[circle,fill=black,minimum size=4pt,inner sep=0pt,outer sep=0pt]{};
}

\def \addlayer{%
  \addtocounter{sncolumncounter}{1}
}

\def \nextlayer{%
  \draw [dashed] (\sncolwidth*\value{sncolumncounter}+\sncolwidth,0.6)--(\sncolwidth*\value{sncolumncounter}+\sncolwidth,\value{snrowcounter}+0.6);
  \addtocounter{sncolumncounter}{2}
}

\newenvironment{sortingnetwork}[2]
{
  \setcounter{sncolumncounter}{0}
  \setcounter{snrowcounter}{#1}
  \def \sn@fullsize{15}
  \begin{tikzpicture}[scale=#2*0.7]
}
{
  \foreach \i in {1, ..., \value{snrowcounter}}
  {
    \draw (-\sncolwidth,\i)--(\sncolwidth*\value{sncolumncounter}+\sncolwidth,\i);
  }
  \end{tikzpicture}
}
\makeatother

\begin{document}

\title{Sorting Networks: the End Game\thanks{%
    Supported by the Israel Science  Foundation, grant 182/13
    and by the Danish Council for Independent Research, Natural Sciences.}
}

\author{Michael Codish\inst{1} \and Lu\'{i}s Cruz-Filipe\inst{2} \and
  Peter Schneider-Kamp\inst{2}}
\authorrunning{M.~Codish, L.~Cruz-Filipe and P.~Schneider-Kamp}
\institute{Department of Computer Science\\
  Ben-Gurion University of the Negev\\
  PoB 653, Beer-Sheva, Israel 84105\\
  \email{mcodish@cs.bgu.ac.il}
  \and Department of Mathematics and Computer Science,\\
  University of Southern Denmark\\
  Campusvej 55, 5230 ODENSE M, Denmark\\
  \email{$\{$lcf,petersk$\}$@imada.sdu.dk}
}

\maketitle
\setcounter{footnote}{0}

\begin{abstract}
  This paper studies properties of the back end of a sorting network
  and illustrates the utility of these in the search for networks of
  optimal size or depth. All previous works focus on properties of the
  front end of networks and on how to apply these to break symmetries
  in the search. The new properties help shed understanding on how
  sorting networks sort and speed-up solvers for both optimal size and
  depth by an order of magnitude.
\end{abstract}

\section{Introduction}
\label{sec:intro}
In the last year, new results were obtained regarding optimality of
sorting networks, concerning both the optimal depth of sorting
networks on $11$ to $16$ channels~\cite{DBLP:conf/lata/BundalaZ14} and
the optimal size of sorting networks on $9$ and $10$
channels~\cite{ourICTAIpaper}.  Both these works apply
symmetry-breaking techniques that rely on analyzing the structure at
the front of a sorting network in order to reduce the number of
candidates to test in an exhaustive proof by case analysis.

In this work, we focus on the dual problem: what does the \emph{end}
of a sorting network look like?  To the best of our knowledge, this
question has never been studied in much detail.
Batcher~\cite{Batcher2011} characterizes a particular class of networks
that can be completed to a sorting network in a systematic way, but his work
only applies to the search for efficient sorting networks.
Parberry~\cite{DBLP:journals/mst/Parberry91} establishes
a necessary condition to avoid examining the last two layers of a candidate
prefix in his proof of optimality of the depth $6$ sorting network
on $9$ channels, but its application requires fixing the previous layers
(although it has similarities to the idea behind our proof of
Theorem~\ref{thm:layerbeforelast} below).

We show that the comparators in the last layer of a sorting
network are of a very particular form, and that the possibilities for
the penultimate layer are also limited.  Furthermore, we show how to control redundancy of a sorting network in a very precise way in
order to restrict its last two layers to a 
significantly smaller number of
possibilities, and we study the impact of this construction in the SAT
encodings used in the proofs of optimality described
in~\cite{DBLP:conf/lata/BundalaZ14,ourICTAIpaper}.

The analysis, results, and techniques in this paper differ substantially from
the work done on the first layers: that work relies heavily on symmetries of
sorting networks to show that the comparators in those layers \emph{may} be
restricted to be of a particular form.  Our results show that the
comparators in the last layers \emph{must} have a particular form.
When working with the first layers it suffices to work up to
renaming of the channels, as there are very general results on how to apply
permutations to the first layers of any sorting network and obtain another
sorting network of the same depth and size.
On the last layer, this is not true: permuting the ending of a sorting network
will not, in general, yield the ending of another sorting network.
We formalize the fact that, as inputs go through a sorting network, the number of channels between pairs of unsorted values gets smaller, until, at the last layer, all occurrences of unsorted pairs of values are on adjacent channels.
To the best of our knowledge, this surprising fact has never been observed before, and it influences the possible positions of comparators in the last layers.
This intuition about the mechanism of sorting networks is formally expressed
by the notion of $k$-block and Theorem~\ref{thm:layerbeforelast}, which is
the main contribution of this paper.

\section{Preliminaries on Sorting Networks}
\label{sec:prelim}
A \emph{comparator network} $C$ with $n$ channels and depth $d$ is a
sequence $C = L_1;\ldots;L_d$ where each \emph{layer} $L_k$ is
a set of comparators $(i,j)$ for pairs of channels $i < j$.
At each layer, every channel may occur in at most one comparator.
The \emph{depth} of $C$ is the number of layers $d$,
and the \emph{size} of $C$ is the total number of comparators in its layers.
If $C_1$ and $C_2$ are comparator networks, then $C_1;C_2$
denotes the comparator network obtained by concatenating the layers of $C_1$ and $C_2$;
if $C_1$ has $m$ layers, it is an \emph{$m$-layer prefix} of $C_1;C_2$.

An input $\bar x\in\{0,1\}^n$ propagates through $C$ as follows:
$\bar x_0 = \bar x$, and for $0<k\leq d$, $\bar x_k$ is the permutation
of $\bar x_{k-1}$ obtained as follows: for each comparator $(i,j)\in L_k$,
the values at positions $i$ and $j$ of $\bar x_{k-1}$ are
reordered in $\bar x_k$ so that the value at position $i$ is not
larger than the value at position $j$.
The output of the network for input $\bar x$ is $C(\bar x)=\bar x_d$, and
$\outputs(C)=\sset{C(\bar x)}{\bar x\in\{0,1\}^n}$.
The comparator network $C$ is a \emph{sorting network} if all elements of
$\outputs(C)$ are sorted (in ascending order).
The zero-one principle~(e.g.~\cite{Knuth73}) implies that a sorting
network also sorts any other totally ordered set, e.g.~integers.

\begin{figure}[t]
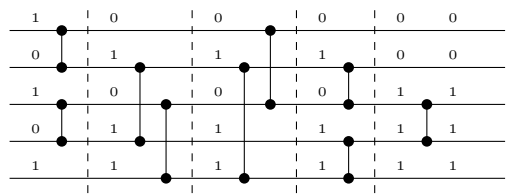

  \centering
  \begin{sortingnetwork}{5}{0.7}
    \nodelabel{\tiny 1,\tiny 0,\tiny 1,\tiny 0,\tiny 1}
    \addcomparator45
    \addcomparator23
    \nextlayer
    \nodelabel{\tiny 1,\tiny 1,\tiny 0,\tiny 1,\tiny 0}
    \addcomparator24
    \addlayer
    \addcomparator13
    \nextlayer
    \nodelabel{\tiny 1,\tiny 1,\tiny 0,\tiny 1,\tiny 0}
    \addcomparator14
    \addlayer
    \addcomparator35
    \nextlayer
    \nodelabel{\tiny 1,\tiny 1,\tiny 0,\tiny 1,\tiny 0}
    \addcomparator12
    \addcomparator34
    \nextlayer
    \nodelabel{\tiny 1,\tiny 1,\tiny 1,\tiny 0,\tiny 0}
    \addcomparator23
    \addlayer
    \nodelabel{\tiny 1,\tiny 1,\tiny 1,\tiny 0,\tiny 0}
  \end{sortingnetwork}
  \caption{An optimal-depth, optimal-size sorting network on $5$ channels, operating
    on the input $10101$.  The channels are numbered from top to bottom, with a comparator
    $(i,j)$ represented as a vertical line between two channels; each comparator moves
    its smallest input to its top channel.  The layers are separated by a vertical dashed
    line.}
  \label{fig:sort-5}
\end{figure}

\emph{Optimal sorting network problems} are about finding the smallest
depth and the smallest size of a sorting network for a given number of
channels $n$.
Figure~\ref{fig:sort-5} shows a sorting network on $5$ channels that has
optimal size ($9$~comparators) and optimal depth ($5$~layers). It also shows how the network sorts
the input $10101$.

In order to determine the minimal depth of an optimal sorting network on $n$~channels,
one needs to consider all possible ways in which such a network can be built.
Parberry~\cite{DBLP:journals/mst/Parberry91} shows that the first layer of
a depth-optimal sorting network on $n$ channels can be assumed to consist of the
comparators $(2k-1,2k)$ for $1\leq k\leq\left\lfloor\frac{n}{2}\right\rfloor$.
Parberry and later Bundala and Z\'avodn\'y pursued the study of the possibilities for the second layer
and demonstrated the impact of this on the search for optimal sorting networks.

The following two observations will be be instrumental for proofs in later sections.
We write $\bar x\leq\bar y$ to denote that every bit of $x$ is less than or equal to the corresponding bit of $y$, and $\bar x<\bar y$ for $\bar x\leq\bar y$ and $x\neq y$.
\begin{lemma}
\label{lem:sorted}
Let $C$ be a comparator network and $\bar x$ be a sorted sequence.
Then $\bar x$ is unchanged by every comparator in $C$.
\end{lemma}

\begin{lemma}[Theorem~4.1 in~\cite{Batcher2011}]
  \label{lem:monot}
  Let $C$ be a comparator network and $\bar x,\bar y\in\{0,1\}^n$ be such that $\bar x\leq\bar y$.
  Then $C(\bar x)\leq C(\bar y)$.
\end{lemma}

\section{The Last Layers of a Sorting Network}
\label{sec:props}
In this section we analyze the last two layers of a sorting network and derive some
structural properties that will be useful both for restricting the search space in
proofs of optimality, and as a tool to understand how a sorting network works.

We begin by recalling the notion of redundant comparator (Exercise 5.3.4.51 of~\cite{Knuth73}, credited to R.L.~Graham).
  Let $C;(i,j);C'$ be a comparator network.
  The comparator $(i,j)$ is \emph{redundant} if $x_i\leq x_j$ for all sequences
  $x_1\ldots x_n\in\outputs(C)$.
%
  If $D'$ is a comparator network obtained 
  by removing every
  redundant comparator from $D$, then $D'$ is a sorting network iff $D$ is a sorting network:
  from the definition it follows that $D(\bar x)=D'(\bar x)$ for every input $\bar x\in\{0,1\}^n$.
%
%
This result was already explored in the proof of optimality of the $25$-comparator sorting network on $9$~channels~\cite{ourICTAIpaper}.
We will call a sorting network without redundant comparators \emph{non-redundant}.
In this section we 
focus on non-redundant sorting networks.

\begin{lemma}
  \label{lem:lastlayer}
  Let $C$ be a non-redundant sorting network on $n$ channels.
  Then all comparators in the last layer of $C$ are of the form $(i,i+1)$.
\end{lemma}

\begin{wrapfigure}[11]{r}{0.2\textwidth}
  \begin{tabular}{cc}
    \begin{sortingnetwork}{3}{0.7}
      \nodelabel{0,0,1}
      \addcomparator13
      \addlayer
      \nodelabel{1,0,0}
    \end{sortingnetwork}
    \\
    $(a)$\\[.5ex]
    \begin{sortingnetwork}{3}{0.7}
      \nodelabel{0,1,1}
      \addcomparator13
      \addlayer
      \nodelabel{1,1,0}
    \end{sortingnetwork}
    \\
    $(b)$
  \end{tabular}

\end{wrapfigure}\ \vspace*{-1.8em}
\begin{proof}
  Let $C$ be as in the premise with a comparator $c=(i,i+2)$ in the last layer.
  We can assume it is the last comparator.
  Since $c$ is not redundant, there is an input $\bar x$ such that channels $i$ to $i+2$ before applying $c$ look like~$(a)$ or~$(b)$ on the right.

  Suppose $\bar x$ is a word yielding case~$(a)$, and let $\bar y$ be any word obtained by replacing one $0$ in $\bar x$ by a $1$.
  Since $C$ is a sorting network, $C(\bar y)$ is sorted, but since $\bar x<\bar y$ the value in channel $i$ before applying $c$ must be a $1$ (Lemma~\ref{lem:monot}), hence $\bar y$ yields situation~$(b)$.
  Dually, given $\bar y$ yielding $(b)$, we know that any $\bar z$ obtained by replacing one $1$ in $\bar y$ by a $0$ will yield $(a)$.


  Thus all inputs with the same number of zeroes as $\bar x$ or $\bar y$ must yield either~$(a)$ or~$(b)$, in particular sorted inputs, contradicting 
  Lemma~\ref{lem:sorted}.
  The same reasoning shows that $c$ cannot have the form $(i,i+k)$ with $k>2$, thus it has to be of the form $(i,i+1)$.\qed
\end{proof}

\begin{corollary}
  \label{cor:lastlayer}
  Suppose that $C$ is a sorting network with no redundant comparators that contains a comparator $(i,j)$ at layer~$d$,
  with $j>i+1$.
  Then at least one of channels $i$ and $j$ is used in a layer $d'$ with $d'>d$.
\end{corollary}
\begin{proof}
  If neither $i$ nor $j$ are used after layer~$d$, then the comparator $(i,j)$
  can be moved to the last layer without changing the function computed by~$C$.
  By the previous lemma $C$ can therefore not be a sorting network.\qed
\end{proof}

Lemma~\ref{lem:lastlayer} restricts the number of possible comparators in the last layer in a sorting network on $n$ channels to $n-1$, instead of $n(n-1)/2$ in the general case.

\begin{theorem}
  \label{thm:lastlayer}
  The number of possible last layers in an $n$-channel sorting network
  with no redundancy is $L_n=F_{n+1}-1$, where $F_n$ denotes the
  Fibonacci sequence.
\end{theorem}
\begin{proof}
  Denote by $L^+_n$ the number of possible last layers on $n$ channels,
  where the last layer is allowed to be empty (so $L_n=L^+_n-1$).
  There is exactly one possible last layer on $1$ channel, and there are
  two possible last layers on $2$ channels (no comparators or one comparator),
  so $L^+_1=F_2$ and $L^+_2=F_3$.

  Given a layer on $n$ channels, there are two possibilities.
  Either the first channel is unused, and there are $L^+_{n-1}$ possibilities
  for the remaining $n-1$ channels; or it is connected to the second channel,
  and there are $L^+_{n-2}$ possibilities for the remaining $n-2$ channels.
  So $L^+_n=L^+_ {n-1}+L^+_{n-2}$, whence $L^+_n=F_{n+1}$.\qed
\end{proof}

Even though $L_n$ grows quickly, it grows 
slower than the number $G_n$ of possible layers in general~\cite{DBLP:conf/lata/BundalaZ14}; in particular, $L_{17}=2583$, whereas $G_{17}=211{,}799{,}312$.

To move (backwards) beyond the last layer, we introduce an auxiliary notion.
\begin{definition}
  Let $C$ be a depth $d$ sorting network without redundant comparators, and let $k<d$.
  A \emph{$k$-block} of $C$ is a set of channels $B$ such that $i,j\in B$ if and only if there is a sequence of channels $i=x_0,\ldots,x_\ell=j$ where $(x_i,x_{i+1})$ or $(x_{i+1},x_i)$ is a comparator in a layer $k'>k$ of $C$.
\end{definition}
Note that for each $k$ the set of $k$-blocks of $C$ is a partition of the set of channels.

Given a comparator network of depth $d$, we will call its $(d-1)$-blocks simply \emph{blocks} -- so Lemma~\ref{lem:lastlayer} states that a block in a sorting network $C$ is either a channel unused at the last layer of $C$ or two adjacent channels connected by a comparator at the last layer of $C$.

\begin{example}
  Recall the sorting network shown in Figure~\ref{fig:sort-5}.
  Its $4$-blocks, or simply blocks, are $\{1\}$, $\{2\}$, $\{3,4\}$ and $\{5\}$,
  its $3$-blocks are $\{1\}$, $\{2,3,4,5\}$,
  and for $k<3$ there is only the trivial $k$-block $\{1,2,3,4,5\}$.
\end{example}

\begin{lemma}
  \label{lem:k-block}
  Let $C$ be a sorting network of depth $d$ on $n$ channels, and $k<d$.
  For each input $\bar x\in\{0,1\}^n$, there is at most one $k$-block that receives a mixture of $0$s and $1$s as input.
\end{lemma}
\begin{proof}
  From the definition of $k$-block, there is no way for values to move from one $k$-block to another.
  Therefore, if there is an input for which two distinct $k$-blocks receive both $0$s and $1$s as inputs, the output will not be sorted.\qed
\end{proof}

\begin{lemma}
  \label{lem:layerbeforelast}
  Let $C$ be a depth $d$ sorting network on $n$ channels without redundant comparators.
  Then all comparators in layer $d-1$
  connect adjacent blocks of $C$.
\end{lemma}
\begin{proof}
  The proof is similar to that of Lemma~\ref{lem:lastlayer}, but now considering blocks instead of channels.
  Let $c$ be a comparator in layer $d-1$ of $C$ that does not connect adjacent blocks of $C$.
  Since $c$ is not redundant, there must be some input $\bar x$ that provides $c$ with input $1$ on its top channel and $0$ on its bottom channel.
  The situation is depicted
  below,
  where $A$ and $C$ are blocks,
  and $B$ is the set of channels in between.
  According to Lemma~\ref{lem:k-block}, there are five possible cases for $A$, $B$ and $C$,
  depending on the number of $0$s in $\bar x$.\medskip

  \hspace*\fill
    \begin{sortingnetwork}{3}{0.7}
      \nodelabel{0,,1}
      \addcomparator13
      \addlayer
      \nodelabel{C,B,A}
    \end{sortingnetwork}
    \qquad
    \begin{tabular}[b]{c|c|c|c|c|c}
      A\,&\,all $0$\,&\,all $0$\,&\,all $0$\,&\,all $0$\,&\,mixed \\
      B\,&\,all $0$\,&\,all $0$\,&\,mixed\,&\,all $1$\,&\,all $1$ \\
      C\,&\,mixed\,&\,all $1$\,&\,all $1$\,&\,all $1$\,&\,all $1$ \\ \hline
     \,&\,$(a)$\,&\,$(b)$\,&\,$(c)$\,&\,$(d)$\,&\,$(e)$
    \end{tabular}
  \hspace*\fill\medskip

  Suppose that $\bar x$ yields $(a)$.
  By changing the appropriate number of $0$s in $\bar x$ to $1$s, we can find a word $\bar y$ yielding case $(b)$,
  since again by monotonicity of $C$ this cannot bring a $0$ to the top input of $c$.
  Likewise, we can reduce $(e)$ to $(d)$.
  But now we can move between $(b)$, $(c)$ and $(d)$ by changing one bit of the word at a time.
  By Lemma~\ref{lem:monot}, this must keep either the top $1$ input of $c$ or the lower $0$,
  while the other input is kept by the fact that $C$ is a sorting network.
  Again this proves that this configuration occurs for all words with the same number of $0$s,
  which is absurd since it cannot happen for the sorted input.\qed
\end{proof}

Combining this result with Lemma~\ref{lem:lastlayer} we obtain the explicit configurations that can occur in a sorting network.
\begin{corollary}
  \label{cor:layerbeforelast}
  Let $C$ be a depth $d$ sorting network on $n$ channels without redundant comparators.
  Then every comparator $(i,j)$ in layer $d-1$ of $C$ satisfies $j-i\leq 3$.
  Furthermore, if $j=i+2$, then either $(i,i+1)$ or $(i+1,i+2)$ occurs in the last layer;
  and if $j=i+3$, then both $(i,i+1)$ and $(i+2,i+3)$ occur in the last layer.
\end{corollary}

The sorting networks in Figure~\ref{fig:tight} show that the bound $j-i\leq 3$ is tight.
%
\begin{figure}[!tb]
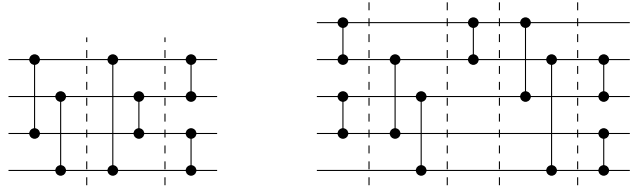

  \centering
  \begin{sortingnetwork}4{0.7}
    \addcomparator24
    \addlayer
    \addcomparator13
    \nextlayer
    \addcomparator14
    \addlayer
    \addcomparator23
    \nextlayer
    \addcomparator12
    \addcomparator34
  \end{sortingnetwork}
  \qquad
  \begin{sortingnetwork}5{0.7}
    \addcomparator45
    \addcomparator23
    \nextlayer
    \addcomparator24
    \addlayer
    \addcomparator13
    \nextlayer
    \addcomparator45
    \nextlayer
    \addcomparator35
    \addlayer
    \addcomparator14
    \nextlayer
    \addcomparator12
    \addcomparator34
  \end{sortingnetwork}
  \caption{Sorting networks containing a comparator $(i,i+3)$ in their penultimate layer.}
  \label{fig:tight}
\end{figure}
We can also state a more general form of Lemma~\ref{lem:layerbeforelast}, proved exactly in the same way.
\begin{theorem}
  \label{thm:layerbeforelast}
  If $C$ is a sorting network on $n$ channels without redundant comparators, then every comparator at layer $k$ of $C$ connects adjacent $k$-blocks of~$C$.
\end{theorem}

When considering the last $n$ comparators instead of the last $k$ layers, induction on $n$ using Theorem~\ref{thm:layerbeforelast} yields the following result.

\begin{corollary}
  \label{cor:blocksize}
  Every $k$-block with $n$ comparators of a sorting network without redundant comparators uses
  at most $n+1$ channels.
\end{corollary}

\section{Co-saturation}
\label{sec:co-sat}
The results in the previous sections allow us to reduce the search space of all possible sorting networks of a given depth simply by identifying necessary conditions on the comparators those networks may have.
However, the successful strategies in~\cite{DBLP:conf/lata/BundalaZ14,ourICTAIpaper,DBLP:journals/mst/Parberry91} all focus on finding \emph{sufficent} conditions on those comparators: identifying a (smaller) set of networks that must contain one sorting network of depth $d$ (or size $k$), if such a network exists at all.

We now follow this idea pursuing the idea of saturation in~\cite{DBLP:conf/lata/BundalaZ14}:
how many (redundant) comparators can we safely add to the last layers of a sorting network?
We will show how to do this in a way that reduces the number of possibilities for the last two layers to a minimum.
Note that we are again capitalizing on the observation that
redundant comparators do not change the function represented by a comparator network and can, thus, be removed or added at will.

\begin{lemma}
  \label{lem:llnf}
  Let $C$ be a sorting network on $n$ channels.
  There is a sorting network $N$ of the same depth whose last layer:
  (i)~only contains comparators between adjacent channels;
  and (ii)~does not contain two adjacent unused channels.
\end{lemma}
\begin{proof}
  We first eliminate all redundant comparators from $C$ to obtain a sorting network $S$.
  By Lemma~\ref{lem:lastlayer} all comparators in the last layer of $S$ are then of the form $(i,i+1)$.
  Let $j$ be such that $j$ and $j+1$ are unused in the last layer of $S$; since $S$ is a sorting network,
  this means that the comparator $(j,j+1)$ is redundant and we can add it to the last layer of $S$.
  Repeating this process for $j=1,\ldots,n$ we obtain a sorting network $N$ that satisfies both desired properties.\qed
\end{proof}
We say that a sorting network satisfying the conditions of Lemma~\ref{lem:llnf} is in \emph{last layer normal form} (llnf).

\begin{theorem}
  \label{thm:llnf}
  The number of possible last layers in llnf on $n$ channels is $K_n=P_{n+5}$,
  where $P_n$ denotes the Padovan sequence, defined as $P_0=1$, $P_1=P_2=0$ and $P_{n+3}=P_n+P_{n+1}$.
\end{theorem}
\begin{proof}
  Let $K^+_n$ be the number of layers in llnf that begin with the comparator $(1,2)$,
  and $K^-_n$ the number of those where channel $1$ is free.
  Then $K_n=K^+_n+K^-_n$.
  Let $n>3$.  If a layer in llnf begins with a comparator, then there are $K_{n-2}$ possibilities for the remaining channels;
  if it begins with a free channel,
  then there are $K^+_{n-1}$ possibilities for the remaining channels.
  Therefore
  $K_n = K^+_n+K^-_n = K_{n-2}+K^+_{n-1} = K_{n-2}+K_{n-3}$.
  There exist one last layer on $1$ channel (with no comparator),
  one on $2$ channels (with one comparator between them) and two on $3$ channels
  (one comparator between either the top two or the bottom two channels),
  so $K_1=P_6$, $K_2=P_7$ and $K_3=P_8$.
  From the recurrence it follows that $K_n=P_{n+5}$.\qed
\end{proof}

Note that $K_n$ grows much slower than the total number $L_n$ of non-redundant last layers identified in Theorem~\ref{thm:lastlayer}.
For example, $K_{17}=86$ instead of $L_{17}=2583$.

If the last layer is required to be in llnf, we can also study the previous layer.
By Lemma~\ref{lem:layerbeforelast},
we know that every block can only be connected to the adjacent ones;
again we can \emph{add} redundant comparators to reduce the number of possibilities for the last two layers.

\begin{lemma}
  \label{lem:twolast}
  Let $C$ be a sorting network of depth $d$ in llnf.
  Let $i<j$ be two channels that are unused in layer $d-1$ and that belong to different blocks.
  Then adding the comparator $(i,j)$ to layer $d-1$ of $C$ still yields a sorting network.
\end{lemma}
\begin{proof}
  Suppose there is an input $\bar x$ such that channel $i$ carries a $1$ at layer $d-1$,
  and channel $j$ carries a $0$ at that same layer.
  Since neither channel is used, their corresponding blocks will receive these values.
  But then $C(\bar x)$ has a $1$ in a channel in the block containing $i$ and a $0$ in the block contaning $j$,
  and since $i<j$ this sequence is not sorted by $C$.
  Therefore the comparator $(i,j)$ at layer $d-1$ of $C$ is redundant, and can be added to this network.\qed
\end{proof}

\begin{lemma}
  \label{lem:onetwo}
  Let $C$ be a sorting network of depth $d$ in llnf.
  Suppose that there is a comparator $(i,i+1)$ in the last layer of $C$,
  that channel $i+2$ is used in layer $d-1$ but not in layer $d$,
  and that channels $i$ and $i+1$ are both unused in layer $d-1$ of $C$ (see Figure~\ref{fig:onetwo}, left).
  Then there is a sorting network $C'$ of depth $d$ in llnf such that channels $i+1$ and $i+2$ are both used in layers $d-1$ and $d$.
\end{lemma}
\begin{proof}
  Since channels $i$ and $i+1$ are unused in layer $d-1$,
  comparator $(i,i+1)$ can be moved to that layer without changing the behaviour of $C$;
  then the redundant comparator $(i+1,i+2)$ can be added to layer $d$,
  yielding the sorting network $C'$ (Figure~\ref{fig:onetwo}, left).
  If $i>1$ and channel $i-1$ is not used in the last layer of $C$,
  then $C'$ must also contain a comparator $(i-1,i)$ in its last layer (Figure~\ref{fig:onetwo}, right).\qed
\end{proof}
\begin{figure}[!tb]
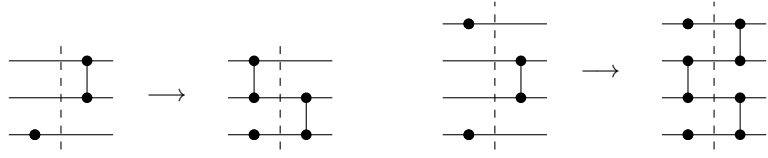

  \hspace*\fill
  \begin{sortingnetwork}{3}{0.7}
    \addcomparator11
    \nextlayer
    \addcomparator23
  \end{sortingnetwork}
  \raisebox{2em}{\ \ $\longrightarrow$}
  \begin{sortingnetwork}{3}{0.7}
    \addcomparator11
    \addcomparator23
    \nextlayer
    \addcomparator12
  \end{sortingnetwork}
  \hfill
  \begin{sortingnetwork}{4}{0.7}
    \addcomparator11
    \addcomparator44
    \nextlayer
    \addcomparator23
  \end{sortingnetwork}
  \raisebox{3em}{\ \ $\longrightarrow$}
  \begin{sortingnetwork}{4}{0.7}
    \addcomparator11
    \addcomparator23
    \addcomparator44
    \nextlayer
    \addcomparator12
    \addcomparator34
  \end{sortingnetwork}
  \hspace*\fill

  \caption{Transformations in the proof of Lemma~\ref{lem:onetwo}.}
  \label{fig:onetwo}
\end{figure}
Lemma~\ref{lem:onetwo} can also be applied if channel $i-1$ (instead of $i+2$) is used at layer $d-1$ and unused in layer $d$.

\begin{definition}
  \label{defn:cosat}
  A sorting network of depth $d$ is \emph{co-saturated} if:
  (i)~its last layer is in llnf,
  (ii)~no two consecutive blocks at layer $d-1$ have unused channels, and
  (iii)~if $(i,i+1)$ is a comparator in layer $d$ and channels $i$ and $i+1$ are unused in layer $d-1$,
  then channels $i-1$ and $i+2$ (if they exist) are used in layer $d$.
\end{definition}


\begin{theorem}
  \label{thm:cosat}
  If $C$ is a sorting network on $n$ channels with depth $d$,
  then there is a co-saturated sorting network $N$ on $n$ channels with depth $d$.
\end{theorem}
\begin{proof}
  Assume $C$ is given.
  Apply Theorem~\ref{thm:llnf} to find a sorting network $S$ in llnf,
  containing no redundant comparators except possibly in the last layer.

  Let $B_1,\ldots,B_k$ be the $(d-1)$-blocks in $S$.
  For $i=1,\ldots,k-1$, if blocks $B_i$ and $B_{i+1}$ have a free channel, add a comparator between them.
  (Note that it may be possible to add \emph{two} comparators between these blocks,
  namely if they both have two channels and none is used in layer~$d-1$.)
  Let $N$ be the resulting network.
  By Lemma~\ref{lem:twolast}, all the comparators added from $S$ to $N$ are redundant,
  so $N$ is a sorting network; by construction, $N$ satisfies~(ii).

  If $N$ does not satisfy~(iii),
  then applying Lemma~\ref{lem:onetwo} transforms it into another sorting network $N'$ that does.\qed
\end{proof}

Table~\ref{tab:co-sat} shows the number of possibilities for the last two layers of a co-saturated sorting network on $n$ channels for $n\leq 17$,
obtained by a representation of these suffixes similar to the one described in~\cite{ourSYNASCpaper}.

\begin{table}[!b]
  \caption{Number of distinct co-saturated two-layer suffixes on $n$ channels, for $n\leq 17$.}
  \label{tab:co-sat}
  \centering
  \begin{tabular}{c|c|c|c|c|c|c|c|c|c|c|c|c|c|c|c}
    $n$\,&\,3\,&\,4\,&\,5\,&\,6\,&\,7\,&\,8\,&\,9\,&\,10\,&\,11\,&\,12\,&\,13\,&\,14\,&\,15\,&\,16\,&\,17 \\ \hline
    \#\,&\,4\,&\,4\,&\,12\,&\,26\,&\,44\,&\,86\,&\,180\,&\,376\,&\,700\,&\,1{,}440\,&\,2{,}892\,&\,5{,}676\,&\,11{,}488\,&\,22{,}848\,&\,45{,}664
  \end{tabular}\smallskip


\end{table}



In the next sections, we show how we can capitalize on these results to improve the proofs of optimal depth and optimal size of sorting networks.

%
%
%
%
%

\section{Implications for Optimal Depth SAT encodings}
\label{sec:implications}
In this section, we describe how SAT encodings in the spirit of~\cite{DBLP:conf/lata/BundalaZ14}
can profit from the results in Sections~\ref{sec:props} and~\ref{sec:co-sat}.
We detail the boolean variables in the model of the encoding, and express our contribution in terms of those.
The remaining details of this construction are immaterial to this paper.
The encoding represents an $n$-channel comparator network of depth $d$ by $d\times n(n-1)$ Boolean variables
\[
  \networkVars^d_n=
       \sset{c^\ell_{i,j}}{1\leq\ell\leq d,1\leq i<j\leq n }
\] 
where the intention is that $c^\ell_{i,j}$ is true if and only if the
network contains a comparator between channels $i$ and $j$ at depth
$\ell$.
Further, to facilitate a concise and efficient encoding of our new results,
we introduce an additional set of $d\cdot n$ Boolean variables capturing which
channels are ``used'' at  a given layer
\[
  \usedVars^d_n=
       \sset{u^\ell_k}{1\leq\ell\leq d,1\leq k\leq n }
\]
where the intention is that $u^\ell_k$ is true if and only if there is
some comparator on channel $k$ at level $\ell$. 
%
%
Using these variables, previous work describes how
the search for an $n$-channel sorting network of depth $d$ is encoded by a formula $\varphi_0$ satisfiable if and only if there is such a network.
If $\varphi_0$ is satisfiable, the network found can be reconstructed from the assignment of the variables $\networkVars^d_n$.

\subsection{Encoding Necessary Conditions}

The results of Section~\ref{sec:props} represent necessary conditions for non-redundant sorting networks. Thus, we can just add them to the SAT encoding as further restrictions of the search space without losing solutions. We start by looking at the last layer, i.e., the layer at depth $d$, and then continue to consider 
layer $d-1$.


Consider first Lemma~\ref{lem:lastlayer}, which states that non-redundant comparators in the last layer have to be of the form $(i,i+1)$. Seen negatively, we can simply forbid all comparators $(i,j)$ where $j > i+1$, that connect non-adjacent channels. This restriction can be encoded straightforwardly by adding the following $(n-1) (n-2)$ unit clauses $\varphi_1$ to the SAT encoding:
\[\varphi_1 = \sset{\neg c^d_{i,j}}{1\leq i,i+1<j\leq n}\]

The restriction from Lemma~\ref{lem:lastlayer} is generalized by Corollary~\ref{cor:lastlayer}, which states that whenever a comparator at any layer connects two non-adjacent channels, necessarily one of these channels is used at a later layer. Similarly to $\varphi_1$ we can encode this by adding one clause for each of the $(n-1)(n-2)/2$ non-adjacent comparator at any given depth $\ell$ using $\varphi_1(\ell)$:
\[\varphi_1(\ell) = \sset{c^\ell_{i,j}\to\bigvee_{\ell <k\leq d}u^k_i\vee u^k_j}{1\leq i,i+1<j\leq n}\]
Note that indeed $\varphi_1(d) = \varphi_1$, as there is no depth $k$ with $\ell < k \leq d$.



We now move on to consider the penultimate layer $d-1$. According to Corollary~\ref{cor:layerbeforelast}, no comparator at this layer can connect two channels more than $3$ channels apart. Similar to Lemma~\ref{lem:lastlayer}, we encode this restriction by adding unit clauses for each of the $(n-3)(n-4)/2$ comparators more than $3$ channels apart:
\[ \varphi_2 = \sset{\neg c^{d-1}_{i,j}}{1\leq i,i+3<j\leq n} \]

Corollary~\ref{cor:layerbeforelast} also states that the existence of a comparator $(i,i+3)$ on the penultimate layer implies the existence of the two comparators $(i,i+1)$ and $(i+2,i+3)$ on the last layer. This is straightforwardly encoded using additional $2(n-3)$ implication clauses:
\[ \varphi_3 = \sset{c^{d-1}_{i,i+3} \to c^d_{i,i+1}\right) \wedge \left(c^{d-1}_{i,i+3} \to c^d_{i+2,i+3}}{1\leq i\leq n-3}\]

Finally, Corollary~\ref{cor:layerbeforelast} also states that the existence of a comparator $(i,i+2)$ on the penultimate layer implies the existence of either of the comparators $(i,i+1)$ or $(i+1,i+2)$ on the last layer. This can be encoded using $n-2$ clauses:
\[ \varphi_4 = \sset{c^{d-1}_{i,i+2}\to c^d_{i,i+1}\vee c^d_{i+1,i+2}}{1\leq i\leq n-2} \]
Empirically, we have found that using $\varphi = \varphi_0 \wedge \varphi_1 \wedge \varphi_2 \wedge \varphi_3 \wedge \varphi_4$ instead of just $\varphi_0$ decreases SAT solving times dramatically. In contrast, adding $\varphi_1(\ell)$ for $\ell < d$ has not been found to have a positive impact.

\subsection{Symmetry Breaking using Sufficient Conditions}

The restrictions encoded so far were necessary conditions for non-redundant sorting networks. In addition, we can break symmetries by using the sufficient conditions from Section~\ref{sec:co-sat}, essentially forcing the SAT solver to \emph{add redundant comparators}.

According to Lemma~\ref{lem:llnf} (ii) we can break symmetries by requiring that there are no adjacent unused channels in the last layer, i.e., that the network is in llnf.
\[ \psi_1 = \sset{u^d_i\vee u^d_{i+1}}{1\leq i<n} \]
Essentially, this forces the SAT solver to add a (redundant) comparator between any two adjacent unused channels on the last layer.

The next symmetry break is based on a consideration of two adjacent blocks.
There are four possible cases: two adjacent comparators, a comparator followed by an unused channel, an unused channel followed by a comparator, and two unused channels.
The latter is forbidden by the symmetry break $\psi_1$ (and thus not regarded further).

The case of two adjacent comparators is handled by formula $\psi_2^a$:
\[ \psi_2^a = \sset{c^d_{i,i+1}\wedge c^d_{i+2,i+3}\to\left(u^{d-1}_i\wedge u^{d-1}_{i+1}\right)\vee\left(u^{d-1}_{i+2}\wedge u^{d-1}_{i+3}\right)}{1\leq i\leq n-3} \]
This condition essentially forces the SAT solver to add a (redundant) comparator on layer $d-1$, if both blocks have an unused channel in that layer.

The same idea of having to add a comparator at layer $d-1$ is enforced for the two remaining cases of a comparator followed by an unused channel or its dual by $\psi_2^b$ and $\psi_2^c$, respectively:
\begin{align*}
  \psi_2^b &= \sset{c^d_{i,i+1}\wedge\neg u^d_{i+2}\to\left(u^{d-1}_i\wedge u^{d-1}_{i+1}\right)\vee u^{d-1}_{i+2}}{1\leq i\leq n-2}\\
  \psi_2^c &= \sset{\neg u^d_i\wedge c^d_{i+1,i+2}\to u^{d-1}_i\vee\left(u^{d-1}_{i+1}\wedge u^{d-1}_{i+2}\right)}{1\leq i\leq n-2}
\end{align*}

The final symmetry break is based on Lemma~\ref{lem:onetwo}, i.e., on the idea of moving a comparator from the last layer to the second last layer. We encode that such a situation cannot occur, i.e., that whenever we have a comparator on the last layer $d$ following or followed by an unused channel, one of the channels of the comparator is used on 
layer $d-1$:
\begin{align*}
  \psi_3^a &= \sset{c^d_{i,i+1}\wedge\neg u^d_{i+2} \to u^{d-1}_{i} \vee u^{d-1}_{i+1}}{1\leq i\leq n-2}\\
  \psi_3^b &= \sset{c^d_{i,i+1}\wedge\neg u^d_{i-1} \to u^{d-1}_{i} \vee u^{d-1}_{i+1}}{2\leq i\leq n-1}
\end{align*}

\begin{table*}[!t]
\caption{SAT solving for $n$-channel, depth $8$ sorting networks with $|R_n|$ $2$-layer filters.
  The table shows the impact of the restrictions on the last two layers in the size of the encoding and the solving time (in seconds) for the slowest unsatisfiable instance,
  as well as the total time for all $|R_n|$ instances.}
  \label{tab:sat}
\centering
\small
\begin{tabular}{|r|r||r|r|r|r||r|r|r|r||}
\cline{3-10} \multicolumn{2}{c||}{}
    &\multicolumn{4}{c||}{{\bf unrestricted} last layer: $\varphi_0$}
    &\multicolumn{4}{c||}{{\bf restricted} last layer: $\psi$}\\
\cline{3-10} \multicolumn{2}{c||}{}
    &\multicolumn{3}{c|}{slowest instance}
    &\multicolumn{1}{c||}{total}
    &\multicolumn{3}{c|}{slowest instance}
    &\multicolumn{1}{c||}{total}\\
\cline{1-5}\cline{7-9}
$n$& $|R_n|$ &
\multicolumn1{c|}{\#clauses} &
\multicolumn1{c|}{\#vars} &
\multicolumn1{c|}{time} &
\multicolumn1{c||}{time} &
\multicolumn1{c|}{\#clauses} &
\multicolumn1{c|}{\#vars} &
\multicolumn1{c|}{time} &
\multicolumn1{c||}{time}\\
\hline
15 & 262 & 278312 & 18217 & 754.74 & {\bf 130551.42}
                      & 335823 & 25209 & 148.35 & {\bf 19029.26}\\
16 & 211 & 453810 & 27007 & 1779.14 & {\bf 156883.21}
                      & 314921 & 22901 & 300.07 & {\bf 24604.53}\\
\hline
\end{tabular}\bigskip
\end{table*}

Empirically, we found that $\psi = \varphi \wedge \psi_1 \wedge \psi_2^a \wedge \psi_2^b \wedge \psi_2^c \wedge \psi_3^a \wedge \psi_3^b$ further improves the performance of SAT solvers.
In order to show optimality of the known depth~$9$ sorting networks on $15$ and $16$ channels,
it is enough to show that there is no sorting network on those numbers of channels with a depth of~$8$.
Previous work~\cite{ourSYNASCpaper} introduces the notion of complete set of prefixes: a set $R_n$ such that if there exists a sorting network on $n$ channels with depth $d$, then there exists one extending a prefix in $R_n$.
Using this result, it suffices to show that there are no sorting networks of depth~$8$ that extend an element of $R_{15}$ or $R_{16}$.
Table~\ref{tab:sat} shows the improvement of using $\psi$ instead of $\varphi_0$,
detailing for both cases the number of clauses, the number of variables and the time to solve the slowest of the $|R_n|$ instances (which are solved in parallel).
We also specify the total solving time (both compilation and SAT-solving) for all $|R_n|$ instances together.
The new encodings are larger per same instance (the slowest instances, showed in the table, are different), but, as indicated in the table, the total time required in order to show that the formulas are unsatisfiable is reduced by a factor of around~$6$.


\section{Conclusion}
\label{sec:conclusion}
This paper presents the first systematic exploration of what happens at the \emph{end} of a sorting network,
as opposed to at the beginning.
We present properties of the last layers of sorting networks.
%
In order to assess the impact of our contribution,
we show how to integrate them into SAT encodings that search for sorting networks of a given depth~\cite{DBLP:conf/lata/BundalaZ14}.
Here, we see an order of magnitude improvement in solving times,
bringing us closer to being able to solve the next open instance of the optimal depth problem ($17$ channels).

While the paper presents detailed results on the end of sorting networks in the context of
proving optimal depth of sorting networks,
the necessary properties of the last layers can also be used to
prove optimal size.
We experimented on adding constraints similar to those in Section~\ref{sec:implications} for the last three comparators, as well as constraints encoding Corollary~\ref{cor:blocksize}, to the SAT encoding presented in~\cite{ourICTAIpaper}. Preliminary results based on uniform random sampling of more than 10\% of the cases indicate that we can reduce the total computational time used in the proof that $25$ comparators are optimal for $9$ channels from $6.5$ years
to just over $1.5$ years. On the $288$-thread cluster originally used for that proof, this corresponds to reducing the actual execution time from over $8$ days to just $2$ days.

These results can also be used to improve times for the \emph{search} for sorting networks.
In a recent paper~\cite{MullerPR14}, the authors introduce an incremental approach to construct sorting networks (iterating between two different SAT problems).
They show that, using the first three layers of a Green filter~\cite{Coles12}, their approach finds a sorting network with $17$~channels and depth $10$, thus improving the previous best upper bound on the depth of a $17$-channel sorting network.
Using the same prefix, together with the constraints on the last two layers described in Section~\ref{sec:implications}, we can find a depth $10$ sorting network in under one hour of  computation.
Without these last layer constraints, this procedure times out after $24$ hours.


\begin{thebibliography}{1}
\providecommand{\url}[1]{\texttt{#1}}
\providecommand{\urlprefix}{URL }

\bibitem{Batcher2011}
Baddar, S.W.A.H., Batcher, K.E.: Designing Sorting Networks: A New Paradigm.
  Springer (2011)

\bibitem{DBLP:conf/lata/BundalaZ14}
Bundala, D., Z{\'a}vodn{\'y}, J.: Optimal sorting networks. In: LATA 2014.
  LNCS, vol. 8370, pp. 236--247. Springer (2014)

\bibitem{ourICTAIpaper}
Codish, M., Cruz-Filipe, L., Frank, M., Schneider-Kamp, P.: Twenty-five
  comparators is optimal when sorting nine inputs (and twenty-nine for ten).
  In: Proceedings of ICTAI 2014. IEEE (2014), accepted for publication

\bibitem{ourSYNASCpaper}
Codish, M., Cruz-Filipe, L., Schneider-Kamp, P.: The quest for optimal sorting
  networks: Efficient generation of two-layer prefixes. In: Proceedings of
  SYNASC 2014. IEEE (2014), accepted for publication

\bibitem{Coles12}
Coles, D.: Efficient filters for the simulated evolution of small sorting
  networks. In: Proceedings of {GECCO}'12. pp. 593--600. {ACM} (2012)

\bibitem{Knuth73}
Knuth, D.E.: The Art of Computer Programming, Volume III: Sorting and
  Searching. Addison-Wesley (1973)

\bibitem{MullerPR14}
M{\"{u}}ller, M., Ehlers, T.: Faster sorting networks for 17, 19 and 20 inputs.
  CoRR  abs/1410.2736 (2014), \url{http://arxiv.org/abs/1410.2736}

\bibitem{DBLP:journals/mst/Parberry91}
Parberry, I.: A computer-assisted optimal depth lower bound for nine-input
  sorting networks. Mathematical Systems Theory  24(2),  101--116 (1991)

\end{thebibliography}

\end{document}